\documentclass[11pt]{article}

\usepackage{times}

\def\colorful{0}

\usepackage{amsthm,amsfonts,amsmath,amssymb,epsfig,color,float,graphicx,verbatim, enumitem}
\usepackage{multirow}
\usepackage[ruled,vlined]{algorithm2e}

\newif\ifhyper\IfFileExists{hyperref.sty}{\hypertrue}{\hyperfalse}
\hypertrue
\ifhyper\usepackage{hyperref}\fi

\usepackage{enumitem}

\usepackage{framed}
\usepackage{nicefrac}

\def\nnewcolor{0}
\ifnum\nnewcolor=1
\newcommand{\nnew}[1]{{\color{red} #1}}
\fi
\ifnum\nnewcolor=0
\newcommand{\nnew}[1]{#1}
\fi

\ifnum\colorful=1
\newcommand{\new}[1]{{\color{red} #1}}

\else
\newcommand{\new}[1]{{#1}}

\fi

\newtheorem{theorem}{Theorem}[section]

\newtheorem{lemma}[theorem]{Lemma}
\newtheorem{informal theorem}[theorem]{Theorem (informal statement)}

\newtheorem{claim}[theorem]{Claim}
\newtheorem{fact}[theorem]{Fact}

\theoremstyle{definition}
\newtheorem{definition}[theorem]{Definition}
\newcommand{\eqdef}{\stackrel{{\mathrm {\footnotesize def}}}{=}}

\newcommand{\mini}[1]{\mbox{minimize} & {#1} &\\}
\newcommand{\maxi}[1]{\mbox{maximize} & {#1 } & \\}
\newcommand{\st}{\mbox{subject to} }
\newcommand{\con}[1]{&#1 & \\}

\newenvironment{lp}{\begin{equation}  \begin{array}{lll}}{\end{array}\end{equation}}
\newenvironment{lp*}{\begin{equation*}  \begin{array}{lll}}{\end{array}\end{equation*}}

\newcommand{\dd}{\mathrm{d}}

\newcommand{\diam}{\mathrm{diam}}

\newcommand{\R}{\mathbb{R}}

\newcommand{\Z}{\mathbb{Z}}

\newcommand{\E}{\mathbf{E}}
\newcommand{\eps}{\epsilon}
\newcommand{\dtv}{d_{\mathrm TV}}

\newcommand{\poly}{\mathrm{poly}}

\newcommand{\KL}{\mathrm{KL}}

\newcommand{\D}{\mathbf{D}}

\newcommand{\mle}{\mathop{\widehat{f}_n}}

\newcommand{\vol}{\mathrm{vol}}

\newcommand{\ch}{\mathrm{Conv}}

\title{A Polynomial Time Algorithm for Maximum Likelihood Estimation of Multivariate Log-concave Densities}

\author{
Ilias Diakonikolas\thanks{Supported by NSF Award CCF-1652862 (CAREER) and a Sloan Research Fellowship.}\\
University of Southern California\\
{\tt diakonik@usc.edu}\\
\and
Anastasios Sidiropoulos\thanks{Supported by NSF Award  CCF-1453472 (CAREER) and NSF grants CCF-1423230 and CCF-1815145.}\\
University of Illinois at Chicago\\
{\tt sidiropo@uic.edu}\\
\and
Alistair Stewart\\ University of Southern California\\
{\tt stewart.al@gmail.com}
}

\begin{document}

\maketitle




\begin{abstract}
We study the problem of computing the maximum likelihood estimator (MLE) of multivariate log-concave densities.
Our main result is the first computationally efficient algorithm for this problem. In more detail, 
we give an algorithm that, on input a set of $n$ points in $\R^d$ and an accuracy parameter $\eps>0$, 
it runs in time $\poly(n, d, 1/\eps)$, and outputs a log-concave density that with high probability 
maximizes the log-likelihood up to an additive $\eps$. Our approach relies on a natural convex optimization 
formulation of the underlying problem that can be efficiently solved by a projected stochastic subgradient method. 
The main challenge lies in showing that a stochastic subgradient of our objective function can be efficiently approximated. 
To achieve this, we rely on structural results on approximation of log-concave densities and 
leverage classical algorithmic tools on volume approximation of convex bodies and uniform sampling 
from convex sets.
\end{abstract}

\thispagestyle{empty}
\setcounter{page}{0}

\newpage

\section{Introduction} \label{sec:intro}

This paper is concerned with the problem of computing the maximum likelihood estimator of multivariate
log-concave densities. Before we state our results,
we provide some background and motivation.

\subsection{Background}

A distribution on $\R^d$ is log-concave if the logarithm 
of its probability density function is concave. More formally, 
we have the following definition:

\begin{definition}[Log-concave Density] \label{def:lc}
A probability density function $f : \R^d \to \R_+$, $d \in \Z_+$, is called {\em log-concave}
if there exists an upper semi-continuous concave function $\phi: \R^d \to [-\infty, \infty)$
such that $f(x) = e^{\phi(x)}$ for all $x \in \R^d$.
We will denote by $\mathcal{F}_d$ the set of upper semi-continuous,
log-concave densities with respect to the Lebesgue
measure on $\R^d$.
\end{definition}

Log-concave densities form a broad nonparametric family
encompassing a wide range of fundamental distributions, including
the uniform, normal, exponential, logistic, extreme value,
Laplace, Weibull, Gamma, Chi and Chi-Squared, and Beta distributions (see, e.g.,~\cite{BagnoliBergstrom05}).
Log-concave probability measures have been extensively investigated in several scientific disciplines, 
including economics, probability theory and statistics, computer science, 
and geometry (see, e.g.,~\cite{Stanley:89, An:95, LV07, Walther09, SW14-survey}). 
The problem of {\em density estimation} for log-concave distributions is of central importance
in the area of non-parametric estimation (see, e.g., ~\cite{Walther09, SW14-survey, Sam17-survey})
and has received significant attention during the past decade in statistics~\cite{CSS10, 
DumbgenRufibach:09, DossW16, ChenSam13, KimSam16, BalDoss14, HW16} 
and theoretical computer science~\cite{CDSS13, CDSS14, ADLS17, CanonneDGR16, DKS16-proper-lc,
DiakonikolasKS17-lc, CDSS18}.

In this work, we focus on the problem of computing the maximum likelihood estimator (MLE)
of a multivariate log-concave density. Formally, we study the following algorithmic question: 
\begin{center}
{\em Is there a computationally efficient algorithm to compute the log-concave MLE on $\R^d$?}
\end{center}
We believe that obtaining an understanding of the above algorithmic 
question is of interest for a number of reasons. 
First, the log-concave MLE is {\em the} prototypical statistical estimator for the class, is fully automatic
(in contrast, e.g., to kernel-based estimators), and has an intriguing geometry~\cite{CSS10, RSU17}. 
Hence, it is of fundamental theoretical and practical importance to understand 
whether it is efficiently computable in multiple dimensions. 
Computing the log-concave MLE is desirable for additional reasons, in particular 
because it satisfies a number of useful properties 
that may not be satisfied by surrogate estimators (see, e.g.,~\cite{Sam17-survey}). 
Finally, we note that an efficient algorithm to compute the MLE 
would yield the first efficient {\em proper} learning algorithm for the class of multivariate log-concave densities.

Recent work~\cite{KimSam16, CDSS18} has nearly characterized 
the rate of convergence (aka sample complexity) of the log-concave MLE with respect 
to the squared Hellinger loss. However, the question of efficient computability has remained open.
Cule, Samworth, and Stewart~\cite{CSS10} (see also~\cite{CuleS10}) 
established several key structural properties of the multivariate log-concave MLE
and proposed a convex formulation to find it. Alas, the work~\cite{CSS10} does not give a polynomial 
time algorithm to compute it. We note that the approach in~\cite{CSS10} 
can be shown to imply an algorithm with runtime $\poly(n^d)$ for our problem,
where $n$ is the sample size and $d$ is the dimension. 
As we explain in Section~\ref{sec:results}, this upper bound is tight for their algorithm, i.e., 
the dependence on the dimension $d$ is inherently exponential. This exponential dependence 
is illustrated in the experimental evaluation 
of the iterative method proposed in~\cite{CSS10}, which does not seem to scale in dimensions more than $4$. 
Recent work by Rathke and Schn{\"o}rr~\cite{RS18} proposed 
a non-convex optimization approach to the problem of computing the log-concave MLE, 
which seems to exhibit faster practical runtimes (scaling to $6$ or higher dimensions). 
Unfortunately however, their method is of a heuristic nature, in the sense that there is no 
guarantee that their solution will converge to the log-concave MLE.


\subsection{Our Results and Techniques} \label{sec:results}

The main result of this paper is an efficient algorithm to
compute the multivariate log-concave MLE.
For concreteness, we formally define the log-concave MLE:

\begin{definition}[Log-concave MLE] \label{def:mle}
Let $x^{(1)}, \ldots, x^{(n)} \in \R^d$.
The log-concave MLE, $\mle = \mle(x^{(1)}, \ldots, x^{(n)})$, 
is the density $\mle \in \mathcal{F}_d$ which maximizes the log-likelihood 
$\ell(f) \eqdef  \sum_{i=1}^n \ln(f(x^{(i)}))$ 
over $f \in \mathcal{F}_d$.
\end{definition}

\noindent As shown in~\cite{CSS10}, 
the log-concave MLE $\mle(x_1, \ldots, x_n)$ exists and is unique. 
Our main result is the first efficient algorithm to compute the log-concave MLE
up to any desired accuracy. 

Before we state our results, we require some terminology:
An evaluation oracle for a distribution with density $f$ is an efficient algorithm
to compute $f(x)$ at any given point $x$. An $\eps$-sampler for a distribution with density 
$f$ is an efficient algorithm that outputs a sample from a distribution with total variation
distance at most $\eps$ from $f$. Our main result is the following:

\begin{theorem}[Main Result] \label{thm:main}
Fix $d \in \Z_+$ and $0<\eps, \tau<1$. There is an algorithm that, on input a set of points $x^{(1)}, \ldots, x^{(n)}$ 
on $\R^d$, and parameters $\eps$, $\tau$, 
runs in $\poly(n, d, 1/\eps, \log(1/\tau))$ time and with probability at least $1-\tau$ 
outputs a (succinct description of a) log-concave density $h^{\ast} \in \mathcal{F}_d$ such that 
$\ell(h^{\ast}) \geq \ell(\mle)-\eps.$ Specifically, the succinct description of our output hypothesis 
admits a $\poly(n, d)$ time evaluation oracle and a $\poly(n, d, 1/\eps)$ 
time $\eps$-sampler for the underlying distribution.
\end{theorem}

Recall that the squared Hellinger loss between two distributions with densities 
$f, g: \R^d \to \R_+$ is $h^2(f, g) = (1/2) \cdot \int_{\R^d} ( \sqrt{f(x)} - \sqrt{g(x)})^2 dx$.
Combined with the known rate of convergence of the log-concave MLE with respect to 
the squared Hellinger loss, Theorem~3 in~\cite{CDSS18}, Theorem~\ref{thm:main} implies the following:

\begin{theorem} \label{thm:global-loss}
Fix $d \in \Z_+$ and $0<\eps, \tau<1$. Let $n = \nnew{\Omega} \left( (d^2/\eps) \ln^3(d/(\eps\tau))  \right)^{(d+3)/2}$. 
There is an algorithm that, given $n$ iid samples from an unknown log-concave density $f_0 \in  \mathcal{F}_d$,
runs in $\poly(n)$ time and outputs a log-concave density $h^{\ast} \in \mathcal{F}_d$ such that
with probability at least $1-\tau$, we have that $h^2(h^{\ast}, f_0) \leq \eps$.
\end{theorem}

We note that Theorem~\ref{thm:global-loss} yields the first efficient proper learning algorithm
for multivariate log-concave densities under a global loss function.

\paragraph{Technical Overview} 

Here we provide a brief overview of our algorithm and its analysis 
in tandem with a comparison to prior work.
Our algorithm proceeds by convex optimization: We formulate the problem of computing
the log-concave MLE of a set of $n$ samples as a convex optimization problem that
we solve via an appropriate first-order method. The main difficulty is that we do not have 
direct access to the (sub-)gradients of the objective function and the naive algorithm
to compute a subgradient at a point takes exponential time. Hence, the key challenge
is how to obtain an efficient algorithm for this task. One of our main contributions is 
a randomized polynomial time algorithm to approximately compute a subgradient of the objective function.
Our algorithm for this task leverages structural results on log-concave densities established in~\cite{CDSS18} 
combined with classical algorithmic results on approximating the volume of convex bodies 
and uniformly sampling from convex sets~\cite{kannan1997random, lovasz2006simulated, lovasz2006hit}.


In more detail, our algorithmic approach leverages a key structural property
of the log-concave MLE, shown in~\cite{CSS10}: The logarithm of the log-concave MLE
$\ln \mle$, is a ``tent'' function, whose parameters are the values $y_1, \ldots, y_n$ 
of the log density at the $n$ input samples $x^{(1)}, \ldots, x^{(n)}$. 
Our convex programming formulation
is very similar, but not identical, to the formulation proposed in~\cite{CSS10}.
Specifically, we seek to maximize the log-likelihood of the probability density function 
obtained by normalizing the \new{log-concave function} whose logarithm 
is the convex hull of the log densities at the samples. 
This objective function is a concave function of the parameters, 
so we end up with a (non-differentiable) convex optimization problem. The crucial observation is that 
the subgradient of this objective at a given point $y$ 
is given by an expectation under the current hypothesis density at $y$.

Given such a convex formulation, we would like to use a first-order method to efficiently 
find an $\eps$-approximate optimum. We note that the objective function is not differentiable everywhere, 
hence we need to work with subgradients. We show that the subgradient of the objective function 
is bounded in $\ell_2$-norm at each point, i.e., the objective function is Lipschitz. 
Another important lemma is that we can restrict the domain of our optimization problem
to a compact convex set of appropriately bounded diameter $D = \poly(n, d)$. 
This is crucial for us, as it allows us to bound the number of iterations of a first order method. 
Given the above, we can in principle use a projected subgradient method to find 
an approximate optimum to our optimization problem,
i.e., find a log-concave density whose log-likelihood is $\eps$-optimal.

It remains to describe how we can efficiently compute a subgradient of our objective function.
Note that the log density \new{of our hypothesis} can be considered as an unbounded convex polytope. 
The previous approach to calculate the subgradient in~\cite{CSS10} 
relied on decomposing this polytope into faces and obtaining a closed form 
for the underlying integral over these faces (that gives their contribution to the subgradient). 
However, this convex polytope is given by $n$ vertices in $d$ dimensions, and 
therefore its number of faces can be $n^{\Omega(d)}$. So, such an algorithm 
cannot run in polynomial time. 

Instead, we note that we can use a linear program to evaluate a function proportional 
to the hypothesis density at a point in time polynomial in $n$ and $d$. To use this oracle 
for the density in order to produce samples from the hypothesis density, 
we use Markov Chain Monte Carlo (MCMC) methods. In particular, we use MCMC 
to draw samples from the uniform distribution on super-level sets and estimate their volumes. 
With appropriate rejection sampling, we can use these samples to obtain 
samples from a distribution that is close to the hypothesis density.  

Since the subgradient of the objective can be expressed as an expectation over this density, 
we can use these samples to sample from a distribution whose expectation is close to a subgradient. 
We then use projected stochastic subgradient descent to find an approximately optimal solution 
to the convex optimization problem. The hypothesis density this method outputs 
has log-likelihood close to the maximum.




\subsection{Related Work} \label{ssec:related}

The general task of estimating a probability distribution under  
qualitative assumptions about the {\em shape} of its probability density function
has a long history in statistics, dating back to the pioneering work of 
Grenander~\cite{Grenander:56} 
who analyzed the maximum likelihood estimator of a univariate monotone density. 
Since then, shape constrained density estimation has been an active research area 
with a rich literature in mathematical statistics and, more recently, in computer science. 
The reader is referred to~\cite{BBBB:72} for a summary of the early work and to~\cite{GJ:14} 
for a recent book on the subject.

The standard method used in statistics for density estimation problems of this form 
is the MLE. See~\cite{Brunk:58, PrakasaRao:69, Wegman:70, 
HansonP:76, Groeneboom:85, Birge:87, Birge:87b,Fougeres:97,ChanTong:04,BW07aos, JW:09, 
DumbgenRufibach:09, BRW:09aos, GW09sc, BW10sn, KoenkerM:10aos, Walther09, 
ChenSam13, KimSam16, BalDoss14, HW16, CDSS18} for a partial list of works analyzing the MLE for various distribution families.
During the past decade, there has been a body of algorithmic work on shape constrained density estimation 
in computer science with a focus on both sample and computational efficiency~\cite{DDS12soda, DDS12stoc, DDOST13focs, CDSS13, CDSS14, CDSS14b, ADHLS15, ADLS17, DKS15, DKS15b, DDKT15, DKS16, DiakonikolasKS17-lc, DiakonikolasLS18}.

Density estimation of log-concave densities has been extensively investigated.
The univariate case is by now well understood~\cite{DL:01, CDSS14, ADLS17, KimSam16, HW16}. 
For example, it is known~\cite{KimSam16, HW16} that $\Theta(\eps^{-5/4})$ samples
are necessary and sufficient to learn an arbitrary log-concave density over $\R$ within squared Hellinger loss $\eps$.
Moreover, the MLE is sample-efficient~\cite{KimSam16, HW16} 
and attains certain adaptivity properties~\cite{KGS16}.  
A line of work in computer science~\cite{CDSS13, CDSS14, ADLS17, CanonneDGR16, DKS16-proper-lc}
gave sample and computationally efficient algorithms for univariate log-concave 
density estimation under the total variation distance.

For the multivariate case, a line of work~\cite{CSS10, DumbgenRufibach:09, DossW16, ChenSam13, BalDoss14} 
has obtained a complete understanding of the global consistency properties of the MLE for any dimension $d$.
Regarding finite sample bounds, Kim and Samworth~\cite{KimSam16} gave 
a sample complexity {\em lower bound} of $\Omega_d \left( (1/\eps)^{(d+1)/2} \right)$ for $d \in \Z_+$ 
that applies to {\em any} estimator, and a near-optimal 
sample complexity {\em upper bound} for the log-concave MLE for $d \leq 3$. 
Diakonikolas, Kane, and Stewart~\cite{DiakonikolasKS17-lc} 
established the first finite sample complexity upper bound for learning multivariate 
log-concave densities under global loss functions. 
Specifically, they proposed an estimator that learns any log-concave density on $\R^d$ 
within squared Hellinger loss $\eps$ with $\tilde{O}_d \left( (1/\eps)^{(d+5)/2} \right)$ samples. 
It should be noted that the estimator proposed in~\cite{DiakonikolasKS17-lc}  
is very different than the log-concave MLE and seems hard 
to compute in multiple dimensions. In recent work, Carpenter {\em et al.}~\cite{CDSS18} 
showed a sample complexity upper bound of $\tilde{O}_d \left( (1/\eps)^{(d+3)/2} \right)$ 
for the multivariate log-concave MLE with respect to squared Hellinger loss, thus obtaining 
the first finite sample complexity upper bound for this estimator in dimension $d\geq 4$. 
Alas, the computational complexity of the log-concave MLE has remained open in
the multivariate case.

Finally, we note that a recent work~\cite{DLS19} obtained a non-proper estimator 
for multivariate log-concave densities with sample complexity $\tilde{O}_d((1/\eps)^{d+2})$ 
and runtime $\tilde{O}_d((1/\eps)^{2d+2})$.

\paragraph{Independent Work.} Contemporaneous work by Axelrod and Valiant~\cite{AV18} 
gives a $\poly(n, d, r, 1/\eps)$ time algorithm to compute an $\eps$-approximation log-concave MLE, 
where $r$ is a parameter bounded by the $\ell_2$-norm of the log-likelihoods of the input points 
$x_1, \ldots, x_n$ under $\mle$.

\subsection{Organization}
After setting up the required preliminaries in Section~\ref{sec:prelims}, in 
Section~\ref{sec:algo} we present our algorithm and an overview of its analysis, 
modulo the efficient sampling procedure we require.
In Section~\ref{sec:sampling}, we describe and analyze our sampling procedure.
Finally, we conclude with a few open problems in Section~\ref{sec:conc}.



\section{Preliminaries} \label{sec:prelims}

\noindent {\bf Notation and Definitions.}
For $m \in \Z_+$, we denote $[m] \eqdef \{1, \ldots, m\}$.
Let $\Delta_n \eqdef \{z = (z_1, \ldots, z_n) \in \R^n: z_i \geq 0, \sum_{i=1}^n z_i = 1 \}$ 
denote the probability simplex on $\R^n$. We will use $\vol(A)$ to denote
the volume of a set $A \subseteq \R^n$, with respect to Lebesgue measure.
A Lebesgue measurable function $f: \R^d \to \R$ is a probability density function (pdf)
if $f(x) \geq 0$ for all $x \in \R^d$ and  $\int_{\R^d} f(x) dx = 1$.

Let $f, g: \R^d \to \R_+$ be probability density functions.
The {\em total variation distance} between $f, g: \R^d \to \R_+$ 
is defined as $\dtv(f, g) = \sup_{S} |f(S) - g(S)|$, where
the supremum is over all Lebesgue measurable subsets of the domain.
We have that $\dtv\left(f, g \right) = (1/2) \cdot \| f -g  \|_1 = (1/2) \cdot \int_{\R^d} |f(x) - g(x)| dx.$
The {\em squared Hellinger distance} between $f, g: \R^d \to \R_+$ is defined as 
$h^2(f, g) = (1/2) \cdot \int_{\R^d} ( \sqrt{f(x)} - \sqrt{g(x)})^2 dx$.
The  {\em Kullback-Leibler (KL) divergence from $g$ to $f$} is defined as
$\KL(f || g) = \int_{-\infty}^{\infty} f(x) \ln \frac{f(x)}{g(x)} dx$.

For $y \in \R_+$ and $f: \R^d \to \R_+$
we denote by $L_f(y) \eqdef \{x \in \R^d \mid f(x) \geq y\}$
its {\em superlevel sets}. If $f$ is log-concave, $L_f(y)$ is a convex set for all $y \in \R_+$.
For a function $f: \R^d \to \R_+$, we will denote by $M_f$ its maximum value.

\smallskip




\section{Log-concave MLE via Convex Optimization} \label{sec:algo}

In this section, we prove Theorem~\ref{thm:main}. At a high-level, we 
first formulate the problem of computing the log-concave 
MLE as a constrained convex optimization problem, and then
show that the latter problem can be solved efficiently with first order methods. 
The main technical challenge is that the subgradients of our convex
objective are not directly accessible, but can be efficiently approximated
using MCMC methods. 

Given a set of $n$ datapoints $x^{(1)}, \ldots, x^{(n)}$ in $\R^d$, we want to compute
a log-concave density that (nearly) maximizes the log-likelihood of the data. A basic
property of the log-concave MLE $\mle$ (see, e.g., Theorem ~2 of~\cite{CSS10} 
or Lemma~15 in~\cite{CDSS18}) is that it is supported on the convex hull of the samples, which we will denote by 
$S_n \eqdef \mathrm{Conv}(\{x^{(i)}\}_{i=1}^n)$. Our approach crucially uses the fact, established in Theorem~2 
of~\cite{CSS10}, that the logarithm of the log-concave MLE, 
$\ln \mle$, is a {\em tent function}, defined as follows:

\begin{definition}[Tent Function] \label{def:tent}
For $y = (y_1, \ldots, y_n) \in \R^n$ and a set of points $x^{(1)}, \ldots, x^{(n)}$ in $\R^d$, 
we define the function $h_y: \R^d \to \R$ as follows: 
\[
h_y(x) = \left\{\begin{array}{ll}
\max \{ z \in \R \textrm{ such that } (x, z) \in \ch(\{(x^{(i)}, y_i)\}_{i=1}^n)\} & \text{ if } x\in S_n\\
 -\infty & \text{ if } x\notin S_n
 \end{array}\right.
\]
\end{definition}

We observe a few basic facts about tent functions. We have that
 $ \min_i  y_i \leq  h_y(x) \leq \max_i y_i$ for all $x \in S_n$.
We note that, for each $y \in \R^n$, $h_y(x)$ is concave as a function of $x$.
Moreover, for fixed $x \in \R^d$, $h_y(x)$ is convex 
(in fact, positive homogeneous, i.e.,~for all $\alpha>0$, $h_{\alpha y}(x)=\alpha h_y(x)$), as a function of $y$.

A basic routine for our overall algorithm is that 
for any $y \in \R^n$, we can efficiently compute the value $h_y(x)$ and 
a subgradient $\partial_y h_y(x)$, for any given $x \in \R^d$. This is established in the 
following simple lemma:

\begin{lemma} \label{lem:h-subgradient-algo}
For any fixed $y \in \R^n$, there is a $\poly(n, d)$ time algorithm that on input
$x \in \R^d$ computes $h_y(x)$ and a subgradient $\partial_y h_y(x)$.
\end{lemma}
\begin{proof}
By Definition~\ref{def:tent}, the value $h_y(x)$ is the optimal objective function value
of the following packing linear program:
\begin{lp} \label{eqn:packing-lp}
\maxi{a \cdot y}
\st \con{a \in \Delta_n, \sum_{i=1}^n a_i x^{(i)} = x}
\end{lp}
Therefore, $h_y(x)$ can be computed in $\poly(n, d)$ time, completing 
the first part of the lemma. 

To establish the second part, 
we note that the set of optimal solutions $a \in \R^n$ to \eqref{eqn:packing-lp} 
corresponds to the set of subgradients $\partial_y h_y(x)$ at point $x$. 
Indeed, if $a^{\ast} \in \R^n$ is an optimal solution to \eqref{eqn:packing-lp}  for $(x,y)$, 
then if we consider the same packing LP corresponding to  $y' \in \R^n$, $y' \neq y$, 
then $a^{\ast}$ is still a feasible solution (since $y$ only appears in the objective), 
and so the optimal value is at least that at $a^{\ast}$. Thus, we have 
$h_{y'}(x) \geq a^{\ast} \cdot y' = h_y(x) + a^{\ast} \cdot (y'-y)$, 
i.e., $a^{\ast}$ is a subgradient of $h_y(x)$ with respect to $y$.
This completes the second part of the lemma. 
In summary, the LP \eqref{eqn:packing-lp} allows us to compute $h_y(x)$ 
and a sub-gradient of $h_y$ at point $x$ in $\poly(n, d)$ time.
\end{proof}

At this point, we establish an important structural property of the log-concave MLE
that will inform our subsequent optimization formulation.
In particular, we show that the ratio between the largest and smallest
values of $\mle$ on the input points $x^{(i)}$ is bounded from above:

\begin{lemma}\label{lem:mle-ratio}
Let $x^{(1)}, \ldots, x^{(n)}$ be a set of points in $\R^d$ and $\mle$ be the corresponding
log-concave MLE. Then, we have that 
\begin{equation} \label{eqn:ratio}
R \eqdef \frac{\max_{i \in [n]}\mle(x^{(i)})}{\min_{i \in [n]}\mle(x^{(i)})} \leq (2 n d)^{2 n d}\;.
\end{equation}
\end{lemma}
\begin{proof}
Let $V = \vol(S_n)$ be the volume of the convex hull of the sample points
and $M =  \max_{x} \mle(x)$ be the maximum pdf value of the MLE. 
By basic properties of the log-concave MLE (see, e.g., Theorem~2 of~\cite{CSS10}), we have that
$\mle(x) >0$ for all $x \in S_n$ and $\mle(x)=0$ for all $x \not \in S_n$.
Moreover, by the definition of a tent function, it follows that $\mle$ attains its global maximum value and its global non-zero positive value
in one of the points $x^{(i)}$.

We can assume without loss of generality that $\mle$ is not the uniform distribution 
on $S_n$, since otherwise $R = 1$ and the lemma follows.
Under this assumption, we have that $R > 1$ or $\ln R > 0$, which implies that $M > 1/V$. 
The following fact bounds the volume of upper level sets 
of any log-concave density:
\begin{fact}[see, e.g., Lemma~8 in ~\cite{CDSS18}] \label{fact:vol}
Let $f \in \mathcal{F}_d$ with maximum value $M_f$. Then for all $w >0$, we have
$\vol(L_{f}(M_{f} e^{-w})) \leq \nnew{{w}^d / M_{f}}$.
\end{fact}
By Fact~\ref{fact:vol} applied to the MLE $\mle$, for $w = \ln R$, we get that
$\vol(L_{\mle}(M/R)) \leq (\ln R)^d/ M$. Since the pdf value of $\mle$ at 
any point in the convex hull $S_n$ is at least that of the smallest sample point 
$x^{(i)}$, i.e., $M/R$, it follows that $S_n$ is contained in $L_{\mle}(M/R)$.
Therefore, 
\begin{equation} \label{eqn:volume-ub}
V \leq (\ln R)^d/M \;. 
\end{equation} 
On the other hand, the log-likelihood of $\mle$
is at least the log-likelihood of the uniform distribution $U_{S_n}$ on $S_n$. 
Since at least one sample point $x^{(i)}$ has pdf value $\mle(x^{(i)}) = M/R$ 
and the other $n-1$ sample points have pdf value $\mle(x^{(i)}) \leq M$, 
we have that 
$$\ln(M/R) + (n-1) \ln M \geq \ell(\mle) \geq \ell(U_{S_n}) = n \ln (1/V) \;,$$
or $n \ln M - \ln R \geq - n \ln V$, 
and therefore $\ln (M V) \geq (\ln R)/n$. 
This gives that
\begin{equation} \label{eqn:volume-lb}
R^{1/n} \leq M V  \;.
\end{equation}
Combining \eqref{eqn:volume-ub} and \eqref{eqn:volume-lb} gives
\begin{equation} \label{eqn:R}
R \leq (\ln R)^{n d}  \;.
\end{equation}
Since $\ln x < x$, $x \in \R$,  setting $x= R^{1 \over{2 n d}}$ gives that
$\ln R < 2 n d \cdot R^{1 \over{2 n d}}$ or 
\begin{equation} \label{eqn:lnR}
(\ln R)^{n d} < (2 n d)^{n d} \cdot R^{1/2}\;.
\end{equation}
By \eqref{eqn:R} and \eqref{eqn:lnR} we deduce that
$R \leq (2 n d)^{n d} \cdot R^{1/2}$ or 
$$R \leq (2 n d)^{2 n d} \;.$$
This completes the proof of Lemma~\ref{lem:mle-ratio}.
\end{proof}

We introduce some notation that will appear throughout this section. 
Let 
\begin{equation} \label{eqn:logR}
B \eqdef 2 n d \cdot \ln (2 n d) = \Theta (n d \log(n d)) \;,
\end{equation}
where it should be noted that $B \geq \ln R$ (by Lemma~\ref{lem:mle-ratio}),
and let $\mathcal{C}$ denote the closed convex set
\begin{equation} \label{eqn:C}
\mathcal{C} \eqdef [-B, 0]^n \;.
\end{equation}
Our optimization formulation to compute the log-concave MLE is as follows:
\begin{lp} \label{eqn:mle-convex}
\mini{F(y) \eqdef  - (\sum_{i=1}^n y_i)/n  + \ln(\int_{x \in S_n} \exp(h_y(x)) \dd x)}
\st \con{y \in \mathcal{C}.}
\end{lp}
Some comments are in order: First, observe that 
$\int_{x \in S_n} \exp(h_y(x)) \dd x) = \int_{x \in \R^d} \exp(h_y(x)) \dd x)$, since $\exp(h_y(x))$
is identically zero outside $S_n$. We note that
the objective function $F$ is similar, but not identical,
to the objective function used in the convex formulation of~\cite{CSS10}. In particular,~\cite{CSS10}
uses the function $\widetilde{F}(y) = - (\sum_{i=1}^n y_i)/n  + \int_{x \in S_n} \exp(h_y(x)) \dd x$,
i.e., without the logarithmic factor in the second term. (While it is not important for our proofs in 
this section, we note that $\widetilde{F}$ automatically guarantees that the optimal solution is a probability density
function, while an optimal solution to \eqref{eqn:mle-convex} 
is a log-concave positive measure that needs to be normalized to give our log-concave density.) 
Second, the feasible set of \eqref{eqn:mle-convex}  is an $\ell_{\infty}$-ball
of relatively small radius (at most $\sqrt{n} B = \poly(n, d)$). It turns out that the polynomial upper 
bound on the radius will be important for us, as the number of iterations of any first order method 
scales linearly with this quantity. As is established formally in Lemma~\ref{lem:correct-formulation}, 
the fact that optimizing over $\mathcal{C}$ 
suffices for our purposes critically relies on Lemma~\ref{lem:mle-ratio}.

The following lemma shows that \eqref{eqn:mle-convex} is a convex optimization problem
and that any near-optimal solution nearly maximizes the log-likelihood of the data:

\begin{lemma} \label{lem:correct-formulation}
We have the following:
\begin{itemize}
\item[(i)] The function $F: \R^n \to \R$ in \eqref{eqn:mle-convex} is convex. 
\item[(ii)] Let $y \in \mathcal{C}$ be such that $F(y)  \leq \min_{y \in \mathcal{C}} F(y)  + \eps/n$.
Then, the log-concave density $H_y(x) = \exp(\widehat{h}_y(x)) = \exp(h_y(x))/\int_{S_n}\exp(h_y(x)) \dd x$
satisfies $\ell(H_y) \geq \ell(\mle)-\eps.$
\end{itemize}
\end{lemma}
\begin{proof}
Part (i) of the lemma follows by iteratively applying known operations that preserve convexity of a function.
Since a sum of convex functions is convex (see, e.g., page 79 of ~\cite{boyd2}), 
it suffices to show that the function $G(y) =  \ln(\int_{S_n} \exp(h_y(x)) \dd x)$ is convex.
Since $h_y(x)$ is a convex function of $y$, by definition, $\exp(h_y(x))$ 
is log-convex as a function of $y$. Since an integral of log-convex functions 
is log-convex (see, e.g., page 106 of ~\cite{boyd2}), it follows that $\int_{S_n} \exp(h_y(x)) \dd x$ is log-convex. 
Therefore, $G$ is convex. We have therefore established that $F$ is convex, as desired.

We proceed to establish part (ii). We note that an optimal solution to 
\eqref{eqn:mle-convex} is not always a probability density function, 
in the sense that it does not necessarily integrate to $1$. 
(Hence, at the end of the optimization procedure, 
we renormalize the log-concave function we obtain.)

We know that $\mle(x) = \exp(h_{y^{\ast}}(x))$, for some $y^{\ast} \in \R^n$.
To prove part (ii), we will show that there exists some scaling of $\mle$
that is optimal for \eqref{eqn:mle-convex}. That is, there exists
$y' \in \mathcal{C}$ such that (a) the normalized version of $\exp(h_{y'}(x))$ 
is identified with $\mle$, and (b) $\exp(h_{y'}(x))$ is optimal for 
\eqref{eqn:mle-convex} . 

By Lemma~\ref{lem:mle-ratio}, $R \leq (2 n d)^{2 n d }$ 
is the ratio between the maximum value and the minimum non-zero value of $\mle$.
Therefore, if we let $y^{\ast} \in \R^n$ be such that $\mle(x) = \exp(h_{y^{\ast}}(x))$, 
then for any $i, j \in [n]$ we have that $|y^{\ast}_i-y^{\ast}_j| \leq \ln R$. 
Now if we set $y'_i = y^{\ast}_i - \max_{j \in [n]} y^{\ast}_j$, 
then we have that $-\ln R \leq y'_i \leq 0$, i.e., $y' \in \mathcal{C}$, where
we used the fact that $\ln R \leq B$. Note that 
$\exp(h_{y'}(x)) =  \exp(h_{y^{\ast}}(x)) \cdot \exp(-\max_j y^{\ast}_j)$, 
and so the normalization of $\exp(h_{y'}(x))$ is the MLE.
That is, $y'$ is a feasible solution to \eqref{eqn:mle-convex} and the corresponding log-concave
function $\exp(h_{y'}(x))$ is a scaling of $\mle$.

We now proceed to show optimality of $y'$. 
Note that \eqref{eqn:mle-convex} can be equivalently written 
as follows:
\begin{lp} \label{eqn:mle-convex2}
\maxi{(\sum_{i=1}^n y_i)/n  - \ln(\int_{x \in S_n} \exp(h_y(x)) \dd x)}
\st \con{y \in \mathcal{C}.}
\end{lp}
Since $h_y(x^{(i)}) \geq y_i$, \eqref{eqn:mle-convex} is a relaxation of the following:
\begin{lp} \label{eqn:mle-convex3}
\maxi{(\sum_{i=1}^n h_y(x^{(i)}))/n  - \ln(\int_{x \in S_n} \exp(h_y(x)) \dd x)}
\st \con{y \in \mathcal{C}.}
\end{lp}
Observe that the objective function of \eqref{eqn:mle-convex3}
is equal to the average log-likelihood of the normalized density obtained from 
$\exp(h_y(x))$. We note that $y'$ is an optimal solution to \eqref{eqn:mle-convex3}.
Indeed, since the objective function is scale-invariant, i.e., does not change if we multiply 
$y$ by any positive number, it follows that the maximum possible value is $(1/n) \ell(\mle)$;
and this value is attained by $y'$.

We now claim that \eqref{eqn:mle-convex2} and \eqref{eqn:mle-convex3} have the same
optimal solutions. Indeed, if $h_y(x^{(i)}) > y_i$, we can improve our objective by replacing 
$y_i$  with $h_y(x^{(i)})$, which does not change $h_y$, so the objective increases.
Therefore, an optimal solution must have $y_i=h_y(x^{(i)})$.
Similarly, if $y$ is an $\eps$-optimal solution, then setting $y_i=h_y(x^{(i)})$ can only 
improve the objective, and so the log-likelihood must be at least as high as the objective.
Since the objective function is equal to the average log-likelihood, it follows that 
an $\eps/n$-optimal solution to \eqref{eqn:mle-convex2}, or equivalently to 
\eqref{eqn:mle-convex}, is an $\eps$-optimal maximizer of the log-likelihood.
\end{proof}




We note that $F$ is not differentiable everywhere, 
hence we need to analyze its subgradients. The following lemma expresses the 
subgradients of $F$ as a (shifted) expectation of an efficiently
computable function with respect to a specific log-concave density:

\begin{lemma} \label{lem:subgradient}
For any $y \in \R^n$, we have that 
\begin{equation} \label{eqn:subgradient}
\partial_y F(y) = (-1/n, \ldots, -1/n) + \E_{X \sim \D_y}[\partial_y h_y(X)] \;,
\end{equation}
where $\D_{y}$ is the log-concave probability distribution defined 
by the log-concave density $\widehat{H}_y(x) = \exp(\widehat{h}_y(x))$
obtained by normalizing the log-concave function $\exp(h_y(x))$.
\end{lemma}

\begin{proof}
Let $G(y) =  \ln(\int_{S_n} \exp(h_y(x)) \dd x)$.
By definition of $F$, we have that 
$$ \partial_y F(y) = (-1/n, \ldots, -1/n) + \partial_y G(y) \;. $$
By the chain rule for subgradients, we can write that
\begin{eqnarray*}
\partial_y G(y) &=& \partial_y \ln \left(\int_{S_n} \exp(h_y(x)) \dd x\right) \\
&=& \frac{\partial_y \int_{S_n} \exp(h_y(x)) \dd x}{\int_{S_n} \exp(h_y(x)) \dd x}  
= \frac{\int_{S_n} \partial_y \exp(h_y(x)) \dd x}{\int_{S_n} \exp(h_y(x)) \dd x} \\
&=& \frac{\int_{S_n} \exp(h_y(x)) \partial_y h_y(x) \dd x}{\int_{S_n} \exp(h_y(x)) \dd x}
= \int \exp(\widehat{h}_y(x)) \partial_y h_y(x) \dd x \\
&=& \E_{X \sim \D_y} [\partial_y h_y(X)] \;,
\end{eqnarray*}
where $\D_y$ is the distribution defined by $\exp(\widehat{h}_y(x))$, 
the normalization of $\exp(h_y(x))$.
\end{proof}

We note that formula~\eqref{eqn:subgradient} is crucial for our approach,
as it enables us to obtain an efficient randomized algorithm to 
approximate a stochastic subgradient of $G$ (and thus of $F$).
For the sake of completeness, before we proceed, 
we recall the necessary background from 
convex optimization that we require.

\paragraph{Convex Optimization Background.}
Consider the constrained convex optimization problem:
\begin{lp} \label{sgd}
\mini{F(y)}
\st \con{y \in \mathcal{C} \;,}
\end{lp}
where $F$ is convex and $\mathcal{C} \subset \R^n$ is a closed convex set.
Recall that a random vector $g \in \R^n$ is called a stochastic subgradient of $F$ at the point $y$ 
if $\E[g] \in \partial_y F(y)$. The stochastic projected subgradient method is as follows:
\begin{itemize}
\item Start from an arbitrary point $y^{(1)} \in \mathcal{C}$.
\item In the $k$-th iteration, $k \geq 1$, compute a stochastic subgradient $g^{(k)}$ at point $y^{(k)}$ and 
        move to the point $y^{(k+1)} := \pi_{\mathcal{C}}(y^{(k)} - a_k g^{(k)}) \;,$ 
        where $a_k$ is the step size and $\pi_{\mathcal{C}}$ is the projection operation on $\mathcal{C}$.
\end{itemize}
We will require a slight strengthening of the following standard result, 
see, e.g., Theorem 3.4.11 in~\cite{Duchi16}:
\begin{fact} \label{thm:sgd}
Let $\mathcal{C}$ be a compact convex set of diameter $\mathrm{diam}(\mathcal{C})<\infty$.
Suppose that  
the projections $\pi_{\mathcal{C}}$ are efficiently computable, and there exists $M <\infty$ such that 
for all $y \in \mathcal{C}$ we have that $\|g\|_2 \leq M$ for all stochastic subgradients.
Then, after $K = \Omega \left( M \cdot \mathrm{diam}(\mathcal{C}) \log(1/\tau)/\eps^2 \right)$ iterations 
of the projected stochastic subgradient method (for appropriate step sizes), 
with probability at least $1-\tau$, we have that 
$F\left(\bar{y}^{(K)}\right) - \min_{y \in \mathcal{C}}F(y) \leq \eps \;,$
where $\bar{y}^{(K)} = (1/K) \sum_{i=1}^K y^{(i)}$.
\end{fact}

We note that Fact~\ref{thm:sgd} assumes that, in each iteration, we can efficiently calculate
an {\em unbiased} stochastic subgradient, i.e., a vector  $g^{(k)}$ such that 
$\E[g^{(k)}] \in \partial_y F(y^{(k)})$.
Unfortunately, this is not the case in our setting,
because we can only {\em approximately} sample from log-concave densities. However, 
it is straightforward to verify that the conclusion of Fact~\ref{thm:sgd} continues to hold
if in each iteration we can compute a random vector $\widetilde{g}^{(k)}$ such that 
$\|\E [\widetilde{g}^{(k)}] - g^{(k)}\|_2 < \delta \eqdef \eps/(2\mathrm{diam}(\mathcal{C}))$, 
for some $g^{(k)} \in \partial_y F(y^{(k)})$. 
This slight generalization is the basic algorithm we use in our setting.

\medskip

We now return to the problem at hand.
As a simple corollary of~\eqref{eqn:subgradient}, we deduce
that the stochastic subgradients of our objective function $F$ are bounded
in $\ell_2$-norm (i.e., $F$ is Lipschitz):

\begin{claim} \label{clm:lip}
The stochastic subgradients of $F: \R^n \to \R$ are bounded from above by $2$,
in $\ell_2$-norm. Therefore, $F$ is $2$-Lipschitz with respect to the $\ell_2$-norm, i.e., 
for all $y, y' \in \R^n$ it holds $|F(y) - F(y')| \leq 2 \|y-y'\|_2$.
\end{claim}
\begin{proof}
By Lemma~\ref{lem:subgradient}
and the triangle inequality, for all $y \in \R^n$ we have that
a stochastic subgradient $g$ of $F$ at point $y$ is of the form
$$g(y) =  (-1/n, \ldots, -1/n) +\partial_y h_y(X)  \;,$$
where $X \sim \D_y$. By Lemma~\ref{lem:h-subgradient-algo}, 
any subgradient $\partial_y h_y(x)$ of $h_y$, for any $x \in \R^d$, 
is  an optimal solution to the LP \eqref{eqn:packing-lp}. It thus 
follows that $\partial_y h_y(x) \in \Delta_n$ and therefore that 
$\|\partial_y h_y(x)\|_2 \leq \|\partial_y h_y(x)\|_1 = 1$. 
That is, for all $y \in \R^n$ we have that $\sup_{x \in S_n} \| \partial_y h_y(x) \|_2 \leq  1$.
The triangle inequality now gives that for all $y \in \R^n$ we have that
$ \| g(y)  \|_2 \leq 1/\sqrt{n} + 1 \leq 2$.

Note that $F(y)$ is continuous as a composition of continuous functions. 
Since any stochastic subgradient of $F$ is bounded by $2$, in $\ell_2$-norm, 
it follows that any subgradient of $F$ is uniformly bounded by $2$.
That is, for all $y \in \R^n$ we have that $ \| \partial_y F(y) \|_2 \leq 2$.
This completes the proof of Claim~\ref{clm:lip}.
\end{proof}

A key technical ingredient of our method is a procedure for sampling 
from the normalization of $\exp(h_y)$, for a given vector $y\in \R^n$.
Furthermore, we need a procedure for computing the normalizing constant needed to obtain a pdf from $\exp(h_y)$.
Our algorithms for both of these tasks are summized in the following:

\begin{lemma}[Efficient Sampling]\label{lem:sampling}
There exist algorithms ${\cal A}_1$ and ${\cal A}_2$ satisfying the following:
Let $\delta, \tau>0$, and $R>1$.
Let $y\in [-\ln R,0]^n$, and let
$\phi$ be the probability density function such that for all $x\in \R^d$, 
\[
\phi(x) = \exp(h_y(x)) / \gamma,
\]
where
\[
\gamma = \gamma(y) = \int_{\R^d} \exp(h_y(x')) \dd x'.
\]
Then the following conditions hold:
\begin{description}
\item{(1)}
On input $y$ and $\delta$, algorithm ${\cal A}_1$ outputs a random vector $X\in \R^d$, 
distributed according to some probability distribution with density $\widetilde{\phi}$, such that 
\[
\|\widetilde{\phi} - \phi\|_1 = O(\delta),
\]
in time
$\poly(n, d, \log R, 1/\delta, \log(1/\tau))$,
with probability at least $1-\tau$.

\item{(2)}
On input $y$ and $\delta$, algorithm ${\cal A}_2$ outputs some $\gamma'>0$, such that 
\[
\gamma / (1+O(\delta)) \leq \gamma'\leq \gamma \cdot (1+O(\delta)),
\]
in time
$\poly(n, d, \log R, 1/\delta, \log(1/\tau))$,
with probability at least $1-\tau$.
\end{description}
\end{lemma}

The proof of Lemma \ref{lem:sampling} is fairly technical and is presented in Section \ref{sec:sampling}.
Given Lemma \ref{lem:sampling}, we now have all the necessary ingredients to describe our main algorithm.
We will solve our constrained convex optimization problem~\eqref{eqn:mle-convex}
using the stochastic projected subgradient method. In particular, we will 
use the slight extension of Fact~\ref{thm:sgd} for the case of approximate stochastic
subgradients. By Claim~\ref{clm:lip}, it follows that we can set $M=2$.
By definition of the convex set $\mathcal{C}$, we have that 
$\mathrm{diam}(\mathcal{C}) \leq \sqrt{n} \cdot B = \tilde{O}(n^{3/2} d)$.
We need to approximate each stochastic subgradient within $\ell_2$-norm error
\begin{equation} \label{eqn:error-sampling}
\delta \eqdef \eps/(2\mathrm{diam}(\mathcal{C})) = \eps/\poly(n, d) \;,
\end{equation}
which we achieve as follows: First, we use the algorithm of Lemma~\ref{lem:sampling}
to compute a random sample $X$ from a distribution $\tilde{\D}_y$ 
of total variation distance at most $\delta$ from $\D_y$. 
We then output
$\partial_y h_y(X)$ by solving the LP \eqref{eqn:packing-lp} (see Lemma~\ref{lem:h-subgradient-algo}).

The pseudo-code of our algorithm is summarized in Algorithm \ref{alg:main}.

\begin{algorithm}[h!]
  \caption{Algorithm to Compute the Log-concave MLE on $\R^d$}
  \label{alg:main}
  \SetKwInOut{Input}{Input}
  \SetKwInOut{Output}{Output}
  \Input{Set of points $\{x^{(i)}\}_{i=1}^n$ on $\R^d$, accuracy parameter $0< \eps < 1$, and confidence parameter $0< \tau < 1$.}
  \Output{A vector $\widetilde{y} \in \R^n$ such that, with probability at least $1-\tau$, $\ell(H_{\widetilde{y}}) \geq \ell(\mle) -\eps$.}
   Use projected stochastic subgradient descent to find an $\eps/n$-approximate optimum of~\eqref{eqn:mle-convex} as follows:\\
   Initially, we set $y^{(1)} \leftarrow 0 \cdot \mathbf{1}^n$ (i.e.~the all-zeros vector).\\
   Let $\delta = \eps/(2\diam({\cal C}))$.\\
   Let $\tau'=\tau/3$.\\
   Let $K=\Omega(M\cdot \diam({\cal C})\log(1/\tau')/\eps^2)$.\\
   For $i=1$ to $K$:\\
   ~~~~~~Let $\phi^{(i)}(x) = \exp(h_{y^{(i)}}(x)) / \int_{\R^d} \exp(h_{y^{(i)}}(x')) \dd x'$, for all $x\in \R^d$.\\
   ~~~~~~Use Lemma \ref{lem:sampling} to sample a random vector $X^{(i)}$ that is distributed according to some $\widetilde{\phi}^{(i)}$, with 
   \[
   \|\widetilde{\phi}^{(i)} - \phi^{(i)}\|_1 < \delta/2,
   \]\\
~~~~~~with probability at least $1-\tau'/K$.\\
~~~~~~Use Lemma \ref{lem:h-subgradient-algo} to compute $z^{(i)}=\partial h_{y^{(i)}}(X^{(i)})$.\\
~~~~~~Set $\widetilde{g}^{(i)} \leftarrow -(1/n)\cdot \mathbf{1}^n + z^{(i)}$.\\
~~~~~~Set $y^{(i+1)} \leftarrow \pi_{{\cal C}}(y^{(i)} - a_i \widetilde{g}^{(i)})$.\\
End-For\\
Let $\bar{y} = (1/K) \sum_{i=1}^K y^{(i)}$.\\
Use Lemma \ref{lem:sampling} to compute some $\gamma'$ that approximates $\gamma$ within relative error $\eps/2$, where
\[
\gamma = \int_{\R^d} \exp(h_{\bar{y}}(x')) \dd x',
\]\\
with probability at least $1-\tau'$.\\
Set $\widetilde{y} \leftarrow \bar{y} - (\ln \gamma') \cdot \mathbf{1}^n$.\\
   \Return{$\widetilde{y}$}
   \end{algorithm}


\begin{proof}[Proof of Theorem~\ref{thm:main}]
It is immediate by the above discussion that Algorithm \ref{alg:main} correctly implements the stochastic projected subgradient method.
By our choice of parameters and Lemma~\ref{lem:correct-formulation}, it follows that 
$\bar{y}$ maximizes the log-likelihood within an additive $\eps/2$. Since the normalization constant $\gamma'$
is accurate within relative error $\eps/2$, it follows that the density corresponding to the normalization of 
$\exp(h_{\bar{y}})$ maximizes the log-likelihood within additive error $\eps$, as required. The efficient evaluation
oracle from our hypothesis follows from the fact that we can efficiently evaluate $h_y(x)$. The approximate sampler
from a close distribution follows directly from Lemma~\ref{lem:sampling}.


Let us now bound the total running time.
Each call to the sampling algorithm from Lemma~\ref{lem:sampling} takes 
time $\poly(n, d, \log R, 1/\delta, \log(1/\tau')) = \poly(n,d,1/\eps,\log(1/\tau))$.
Computing $z^{(1)}$ using Lemma \ref{lem:h-subgradient-algo} takes time $\poly(n,d)$.
Computing the projection $\pi_{\cal C}$ can be done in time $O(n)$, 
simply by iterating over all coordinates of a vector and rounding it to its nearest value in the interval $[-\log R,0]$.
Thus, each iteration of the main loop of Algorithm \ref{alg:main} takes time $\poly(n,d,1/\eps,\log(1/\tau))$.
Since there are $K=\poly(n,d,1/\eps,\log(1/\tau))$ iterations, it follows that the execution 
of the main loop takes time $\poly(n,d,1/\eps,\log(1/\tau))$.
Finally, the computation of $\gamma'$ using Lemma \ref{lem:sampling} takes time $\poly(n,d,1/\eps,\log(1/\tau))$.
We conclude that the total running time of the algorithm is $\poly(n,d,1/\eps,\log(1/\tau))$.

Finally, we bound the failure probability.
There are $K$ invocations of the sampling algorithm given by Lemma \ref{lem:sampling}.
Each one of them fails with probability at most $\tau'/K$.
By the union bound, they all succeed with probability at least $1-\tau'$.
The algorithm for computing $\gamma'$ succeeds with probability at least $1-\tau'$.
If all of these algorithms succeed, by Fact \ref{thm:sgd} and the subsequent discussion 
we have that the overall algorithm succeeds with probability at least $1-\tau'$.
It follows by the union bound that the overall success probability is at least $1-3\tau'=1-\tau$, as required. 
\end{proof}

\section{Efficient Sub-gradient Approximation via Sampling} \label{sec:sampling}

In this Section, we prove Lemma \ref{lem:sampling}, which gives an algorithm from sampling from the log-concave distribution computed by our algorithm.

The main algorithm used in the proof of Lemma \ref{lem:sampling} is concerned mainly with part (1) in its statement.
The pseudocode of this sampling procedure is given in Algorithm \ref{alg:sampling}.
Using the notation from Algorithm \ref{alg:sampling}, part (2) is easier to describe and we thus omit the pseudocode.
There are several subroutines that are used by Algorithm \ref{alg:sampling}: a procedure for approximating the volume of a convex body given by a membership oracle, and a procedure for sampling from the uniform distribution supported on such a body.
For these procedures we use the algorithms by Kannan, L\'{o}vasz and Simonovitz \cite{kannan1997random}, 
which are summarized in Theorems \ref{thm:KLS_volume} and \ref{thm:KLS_sampling} respectively.

\begin{algorithm}[h]
  \caption{Algorithm to sample from the normalized $H_y$}\label{alg:sampling}
  \label{alg:pc-sdp}
  \SetKwInOut{Input}{Input}
  \SetKwInOut{Output}{Output}
  \Input{Set of points $\{x^{(i)}\}_{i=1}^n$ on $\R^d$, vector $y\in [-\ln R, 0]^n$, parameter $0< \delta < 1$.}
  \Output{A random vector $X \in \R^d$ sampled from a probability distribution with density function $\widetilde{\phi}$, such that $\|\widetilde{\phi}-\phi\|_1 \leq \delta$, where $\phi(x)=\exp(h_y(x))/\int_{\R^d} \exp(h_y(x')) \dd x'$.}
Step 1. Let $m = \lceil 1 + \log_2R\rceil$. For any $i\in [m]$ let $L_i=\{x\in \R^d:H_y(x)\geq M_{H_y} 2^{-i}\}$.
 For any $i\in [m]$ compute an estimate $\widetilde{\vol}(L_i)$ of $\vol(L_i)$ such that 
  \[
  \vol(L_i)/(1+\delta)\leq \widetilde{\vol}(L_i) \leq \vol(L_i)(1+\delta).
  \]\\
Step 2. For any $i\in [m]$ let $u_i$ be the uniform probability distribution on $L_i$, and let $\widetilde{u}_i$ be an efficiently samplable probability distribution such that 
  \[
  \|\widetilde{u}_i - u_i\|_1 \leq \delta.
  \]\\
Step 3. Let $\widetilde{c} = \sum_{i=1}^m 2^{-i} \widetilde{\vol}(L_i) + 2^{-m} \widetilde{\vol}(L_m)$.\\
Step 4. Let $\widehat{D}$ be the probability distribution on $[m]$ with 
\[
\Pr_{I\sim \widetilde{D}} [I = i] = \left\{\begin{array}{ll}
\widetilde{\vol}(L_i) \cdot 2^{-i} / \widetilde{c} & \text{ if $i\in \{1,\ldots,m-1\}$}\\
2 \cdot \widetilde{\vol}(L_m) \cdot 2^{-m} / \widetilde{c} & \text{ if $i=m$}
\end{array}\right.
\]\\
Step 5. Sample $I\sim \widetilde{D}$.\\
Step 6. Sample $X \sim \widetilde{u}_I$.\\
Step 7. For any $x\in \R^d$ let
\[
G_y(x) = M_{H_f} 2^{-\lfloor \log_2(M_{H_y}/H_y(x)) \rfloor}
\]\\
Step 8. With probability $1-H_y(X)/G_y(X)$ go to Step 5.\\
\Return{X}
\end{algorithm}

\begin{theorem}[Kannan, Lov\'{a}sz and Simonovits \cite{kannan1997random}]
\label{thm:KLS_volume}
The volume of a convex body $K$ in $\R^d$, given by a membership oracle, can be approximated to within a relative error of $\delta$ with probability $1-\tau$ using
\[
d^5 \cdot \poly(\log d, 1/\delta, \log(1/\tau)) 
\]
oracle calls.
\end{theorem}

\begin{theorem}[Kannan, Lov\'{a}sz and Simonovits \cite{kannan1997random}]
\label{thm:KLS_sampling}
Given a convex body $K\subset \R^d$, with oracle access, and some $\delta>0$, we can generate a random point $u \in K$ that is distributed according to a distribution that is at most $\delta$ away from uniform in total variation distance, using 
\[
d^5 \cdot \poly(\log d, 1/\delta)
\]
oracle calls.
\end{theorem}


In order to use the algorithms in Theorems \ref{thm:KLS_volume} and \ref{thm:KLS_sampling} in our setting, we need a membership oracle for the superlevel sets of the function $H_y$.
Such an oracle can clearly be implemented using a procedure for evaluating $H_y$, 
which is given \new{by Lemma~\ref{lem:h-subgradient-algo}.}
\new{We also need a separation oracle for these superlevel sets, that is given in the following lemma:}

\begin{lemma}[Efficient Separation] \label{lem:sep}
There exists a $\poly(n, d)$ time separation oracle for the superlevel sets of $H_y(x) = \exp(h_y(x))$.
\end{lemma}
\begin{proof}
To construct our separation oracle, we will rely on the covering LP that is dual to \eqref{eqn:packing-lp}.
The dual to \eqref{eqn:packing-lp} looks for the hyperplane that is above all the $(x^{(i)},y_i)$ 
that has minimal $y$ at $x$. More specifically, it is the following LP:
\begin{lp} \label{eqn:covering-lp}
\mini{\beta_0 + \sum_{j=1}^d \beta_j x_j}
\st \con{\beta \in \R^{d+1}, \beta_0 + \sum_{j=1}^d \beta_j x^{(i)}_j \geq y_i, i \in [n].}
\end{lp}
Now suppose that we are interested in a super level set $L_{H_y}(l)$. We can use the above LP 
to compute $h_y(x)$ (and thus $H_y(x)$) and check if it is in the superlevel set. 
Suppose that it is not, then there will be a solution $\beta \in \R^{d+1}$ 
whose value is below $ \ln l$, say $ \ln l-\delta$ for some $\delta>0$. Consider an $x' $ 
in the halfspace $\beta_0 + \sum_{j=1}^d \beta_j x'_j \leq \ln l - \delta/2$ which has $x$ in the interior. 
Since $x$ does not appear in the objective, $\beta$ is a feasible solution for the dual LP 
\eqref{eqn:covering-lp} with $y, x'$, and so $h_y(x') \leq \ln l-\delta/2$, which implies that 
$x'$  is not in the superlevel set. Therefore, $\beta_0 + \sum_j \beta_j x'_j = \ln l - \delta/2$ 
is a separating hyperplane for $x$ and the level set. 
This completes the proof of Lemma~\ref{lem:sep}.
\end{proof}

Given all of the above ingredients, we are now ready to proof the main result of this section.

\begin{proof}[Proof of Lemma \ref{lem:sampling}]
We first prove part (1) of the assertion.
To that end we analyze the sampling procedure described in Algorithm \ref{alg:sampling}.
Recall that 
$m=\lceil 1+\log_2 R\rceil$,
and for any $i\in [m]$, we define the superlevel set
\[
L_i = \{x\in \R^d : H_y(x) \geq M_{H_y} \cdot 2^{-i}\} \;.
\]
For any $x\in \R^d$ recall that 
\[
G_y(x) = M_{H_f} 2^{-\lfloor \log_2(M_{H_y}/H_y(x)) \rfloor} \;.
\]
For any $A\subseteq \R^d$, let $\chi_A:\R^d\to\{0,1\}$ be the indicator function for $A$.
It is immediate that for all $x\in \R^d$,
\begin{align*}
G_y(x) &= M_{H_y} \sum_{i=1}^{\infty} 2^{-i} \chi_{L_i}(x) \\
 &= M_{H_y} \sum_{i=1}^m 2^{-i} \chi_{L_i}(x) + 2^{-m} \chi_{L_i}(m) & \text{(since $H_y(x)=0$ for all $x\notin L_m$)} 
\end{align*}
Let
\[
c = \sum_{i=1}^m 2^{-i} \vol(L_i) + 2^{-m}  \vol(L_m).
\]
We have
\begin{align}
\int_{\R^d} G_y(x) \dd x &= M_{H_y} \left(\sum_{i=1}^m 2^{-i} \vol(L_i) + 2^{-m}  \vol(L_m)\right) = M_{H_y} c. \label{eq:Gy_integral}
\end{align}
Let 
\[
\widehat{G}_y(x) = G_y(x) / (M_{H_y} c).
\]
It follows by \eqref{eq:Gy_integral} that $\widehat{G}_y$ is a probability density function.

Let $D$ be the probability distribution on $\{1,\ldots,m\}$, where
\[
\Pr_{I\sim D}[I = i] = \left\{\begin{array}{ll}
\vol(L_i) \cdot 2^{-i} / c & \text{ if $i\in \{1,\ldots,m-1\}$}\\
2\cdot \vol(L_m) \cdot 2^{-m} / c & \text{ if $i=m$}
\end{array}\right.
\]
For any $i\in [m]$, let $u_i$ be the uniform probability density function on $L_i$.
To sample from $\widehat{G}_y$, we can first sample $I\sim D$, and then sample $X\sim u_I$.


Let $\widehat{H}_y:\R^d \to \R_{\geq 0}$ be the probability density function obtained by normalizing $H_y$; 
that is, for all $x\in \R^d$ let
\[
\widehat{H}_y(x) = H_y(x) / c',
\]
where 
\[
c' = \int_{\R^d} H_y(x) \dd x.
\]
Consider the following random experiment:
first sample $X\sim \widehat{G}_y$, and then accept with probability $H_y(x)/G_y(x)$; conditioning on accepting, the resulting random variable $X\in \R^d$ is distributed according to $\widehat{H}_y$.
Note that since for all $x\in \R^d$, $G_y(x)/2 \leq H_y(x) \leq G_y(x)$, it follows that we always accept with probability at least $1/2$.
Let $\alpha$ be the probability of accepting.
Then 
\begin{align*}
\alpha &= \int_{\R^d} \widehat{G}_y(x) (H_y(x)/G_y(x))  \dd x,
\end{align*}
and thus 
\begin{align}
\int_{\R^d} H_y(x) \dd x &= \int_{\R^d} G_y(x) (H_y(x)/G_y(x)) \dd x  \notag \\
 &= M_{H_y} c \int_{\R^d} \widehat{G}_y(x) (H_y(x)/G_y(x)) \dd x \notag \\ 
 &= M_{H_y} c \alpha \label{eq:mass_Hy} \;.
\end{align}
By Theorem \ref{thm:KLS_volume}, for each $i\in [m]$, we compute an estimate, $\widetilde{\vol}(L_i)$, to  $\vol(L_i)$, to within relative error $\delta$, using $\poly(d, 1/\delta, \log(1/\tau'))$ oracle calls, with probability at least $\tau'$, where $\tau'=\tau/n^b$, for some constant $b>0$ to be determined;
moreover, by Theorem \ref{thm:KLS_sampling}, we can efficiently sample, 
using $\poly(d, 1/\delta)$ oracle calls, from a probability distribution $\widetilde{u}_i$ 
with $\|u_i - \widetilde{u}_i\| \leq \delta$. 
Each of these oracle calls is a membership query in some superlevel set of $H_y$.
This membership query can clearly be implemented 
if we can compute that value $H_y$ at the desired query point $x$, 
which can be done in time $\poly(n,d)$ using lemma~\ref{lem:h-subgradient-algo}.
Thus, each oracle call takes time $\poly(n,d)$.
Let
\begin{align}
\widetilde{c} &= \sum_{i=1}^m 2^{-i} \widetilde{\vol}(L_i) + 2^{-m} \widetilde{\vol}(L_m). \label{eq:c_tilde}
\end{align}
Since for all $i\in [m]$, $\vol(L_i)/(1+\delta)\leq \widetilde{\vol}(L_i) \leq \vol(L_i) (1+\delta)$, it is immediate that 
\[
c / (1+\delta) \leq \widetilde{c}\leq c (1+\delta) \;.
\]
Recall that Algorithm \ref{alg:sampling} uses the probability distribution $\widetilde{D}$ on $[m]$, where
\[
\Pr_{I\sim \widetilde{D}} [I = i] = \left\{\begin{array}{ll}
\widetilde{\vol}(L_i) \cdot 2^{-i} / \widetilde{c} & \text{ if $i\in \{1,\ldots,m-1\}$}\\
2 \cdot \widetilde{\vol}(L_m) \cdot 2^{-m} / \widetilde{c} & \text{ if $i=m$}
\end{array}\right.
\]
Consider the following random experiment, which corresponds to Steps 5--6 of Algorithm \ref{alg:sampling}:
We first sample $I\sim \widetilde{D}$, and then we sample $X\sim \widetilde{u}_I$.
The resulting random vector $X\in \R^d$ is  distributed according to 
\[
 \widetilde{G}_y(x) = \frac{1}{\widetilde{c}} \left ( \sum_{i=1}^m 2^{-i} \widetilde{\vol}(L_i) \widetilde{u}_i(x) + 2^{-m}  \widetilde{\vol}(L_m) \widetilde{u}_m(x) \right).
\]

Next, consider the following random experiment, which captures Steps 5--8 of Algorithm \ref{alg:sampling}:
We sample $X\sim \widetilde{G}_y$, and we accept with probability $H_y(X)/G_y(X)$.
Let $\widetilde{H}_y$ be the resulting probability density function supported on $\R^d$ obtained by conditioning the above random experiment on accepting.
Let $\widetilde{\alpha}$ be the acceptance probability.
We have
\[
 \widetilde{\alpha} =\int_{\R^d} (H_y(x)/G_y(x)) \widetilde{G}(x) \dd x.
\]

We have
\begin{align*}
\|D_i - \widetilde{D}_i\|_1 &= \sum_{i=1}^{m-1} 2^{-i} \cdot \left|\frac{\vol(L_i)}{c}-\frac{\widetilde{\vol}(L_i)}{\widetilde{c}}\right| + 2 \cdot 2^{-m} \cdot \left|\frac{\vol(L_m)}{c}-\frac{\widetilde{\vol}(L_m)}{\widetilde{c}}\right|\\
 &= \sum_{i=1}^{m-1} 2^{-i} \cdot \left|\frac{\vol(L_i)}{c}-\frac{\vol(L_i)(1+\delta)}{c/(1+\delta)}\right| + 2 \cdot 2^{-m} \cdot \left|\frac{\vol(L_m)}{c}-\frac{\vol(L_m)(1+\delta)}{c/(1+\delta)}\right|\\
 &\leq \sum_{i=1}^{m-1} 2^{-i} \frac{\vol(L_i)}{c} 3\delta + 2 \cdot 2^{m} \frac{\vol(L_m)}{c} 3\delta \\
 &= 3\delta.
\end{align*}
It follows that
\begin{align*}
\|\widehat{G}_y - \widetilde{G}_y\|_1 &\leq \|D_i - \widetilde{D}_i\| + \max_i \| u_i-\widetilde{u}_i\|_1 \leq 3\delta + \delta \leq 4 \delta,
\end{align*}
and so
\[
|\alpha-\widetilde{\alpha}| \leq \int_{\R^d} \frac{H_y(x)}{G_y(x)} \left|\widehat{G}_y(x) - \widetilde{G}_y(x)\right| \dd x \leq \int_{\R^d} \left|\widehat{G}_y(x) - \widetilde{G}_y(x)\right| \dd x \leq \|\widehat{G}_y - \widetilde{G}_y \|_1 \leq 4 \delta.
\]
Note that $\widehat{H}_y(x)/\alpha = \widehat{G}_y(x) (H_y(x)/G_y(x) )$ and $\widetilde{H}_y(x)/\widetilde{\alpha} = \widetilde{G}_y(x) (H_y(x)/G_y(x) )$
and so
\begin{align*}
\|\widetilde{H}_y - \widehat{H}_y\|_1 &\leq \alpha \left(\|\widetilde{H}_y/\alpha - \widehat{H}_y/\alpha\|_1 + \|\widehat{H}_y/\widetilde{\alpha} - \widehat{H}_y/\alpha\|_1\right) & \text{(by the triangle inequality)}\\
 &= \alpha \left(\|\widetilde{H}_y/\alpha - \widehat{H}_y/\alpha \|_1 + |1/\widetilde{\alpha} - 1/\alpha|\right) \\
 &= \alpha \int_{\R^d} (H_y(x)/G_y(x)) | \widetilde{G}_y(x) - \widehat{G}_y(x) | + |\alpha - \widetilde{\alpha}|/\widetilde{\alpha} \\
 &\leq \|\widehat{G}_y - \widetilde{G}_y \|_1 + 2 |\alpha - \widetilde{\alpha}| \\
 &\leq 12 \delta,
\end{align*}
which establishes that the random vector $X$ that Algorithm \ref{alg:sampling} 
outputs is distributed according to a probability distribution 
$\widetilde{\phi}$ such that $\|\widetilde{\phi}-\phi\|_1\leq 10\delta$, as required. 

In order to bound the running time, we observe that all the steps of the algorithm 
can be implemented in time $\poly(n, \new{d}, \log R, 1/\delta, \log(1/\tau))$. The most expensive 
operation is approximating the volume of an superlevel set $L_i$ and sampling for $L_i$, 
using Theorems \ref{thm:KLS_volume} and \ref{thm:KLS_sampling}.
By the above discussion, using Lemma \ref{lem:h-subgradient-algo} 
each of these operations can be implemented in time $\poly(n, \new{d}, 1/\delta,\log(1/\tau))$.
The algorithm suceeds if all the invocations of the algorithm of Theorem \ref{thm:KLS_volume} 
are successful; by the union bound, this happens with probability 
at least $1-\tau'\cdot\poly(n) = 1-\tau' n^b \poly(n) \geq 1-\tau$, 
where the inequality follows by choosing some sufficiently large constant $b>0$.
This establishes part (1) of the lemma.

It remains to prove part (2).
By \eqref{eq:mass_Hy} we have that $\gamma = M_{H_y} c \alpha$.
Algorithm ${\cal A}_2$ proceeds as follows.
First, we compute $M_{H_y}$.
By the convexity of $h_y$, it follows that the maximum value of $M_{H_y}$ 
is attained on some sample point $x_i$; that is, 
$M_{H_y} = \max_{i\in [n]} H_y(x_i)$.
Since we can evaluate $H_y$ in polynomial time using 
Lemma \ref{lem:h-subgradient-algo}, it follows that we can also compute $M_{H_y}$ in polynomial time.
Next, we compute $\widetilde{c}$ using formula \ref{eq:c_tilde}.
Arguing as in part (1), this can be done in time $\poly(n, 1/\delta, \log(1/\tau))$, and with probability at least $1-\tau/2$.
Finally, we estimate $\widetilde{\alpha}$.
The value of $\widetilde{\alpha}$ is precisely the acceptance probability of the random experiment described in Steps 5--8 of Algorithm \ref{alg:sampling}.
Since $\alpha\geq 1/2$, and $|\alpha-\widetilde{\alpha}| \leq 4\delta$, it follows that for $\delta<1/16$, we can compute an estimate $\bar{\alpha}$ of the value of $\widetilde{\alpha}$, to within error $1+O(\delta)$, 
with probability at least $1-\tau/2$, after $O(\log(1/\tau))$ repetitions of the random experiment.
The output of algorithm ${\cal A}_2$ is $\gamma'=M_{H_y} \widetilde{c} \bar{\alpha}$.
We obtain that, with probability at least $1-\tau$, we have
\[
\gamma' = M_{H_y} \widetilde{c} \bar{\alpha}  \leq M_{H_y} c (1+\delta) \alpha (1+O(\delta))  = \gamma (1+O(\delta)) \;,
\]
and
\[
\gamma' = M_{H_y} \widetilde{c} \bar{\alpha}  \geq M_{H_y} (c / (1+\delta)) ( \alpha /(1+O(\delta)) ) = \gamma / (1+O(\delta)) \;,
\]
which concludes the proof.
\end{proof}

\section{Conclusions} \label{sec:conc}
In this paper, we gave a $\poly(n, d, 1/\eps)$ time 
algorithm to compute an $\eps$-approximation of the log-concave
MLE based on $n$ points in $\R^d$. Ours is the first algorithm
for this problem with a sub-exponential dependence in the dimension $d$.
We hope that our approach may lead to more practical methods for computing
the log-concave MLE in higher dimensions than was previously possible.

From a theoretical standpoint, an immediate open question is whether there exists
an algorithm for our problem running in time 
$\poly(n, d, \log(1/\eps))$. Since such an algorithm would seem to require going beyond first order methods, 
and the sample size $n$ is typically large (as a function of $1/\eps$), 
it is unclear whether a method with such an asymptotic runtime could be practically viable. An intriguing direction would 
be to explore the complexity of the log-concave MLE under natural restrictions to log-concavity. 
It would also be interesting to understand the complexity of the MLE for mixtures of two
log-concave distributions.
 
 \paragraph{Acknowledgements.} We are grateful to Jelena Diakonikolas for sharing
her expertise in optimization.


\bibliographystyle{alpha}

\bibliography{allrefs}

\appendix

\end{document}